\documentclass[letterpaper]{article} 
\usepackage{relsize,amsfonts}
\usepackage{amsmath,amssymb,amsthm} 
\usepackage{unicode}
\usepackage{fullpage}
\usepackage{hyperref}
\usepackage{xspace}
\usepackage{booktabs}
\usepackage{stmaryrd}
\usepackage{color}
\newtheorem{theorem}{Theorem}
\newtheorem{lemma}{Lemma}

\newtheorem{fact}{Fact}

\newtheorem{definition}{Definition}
\newtheorem{corollary}{Corollary}

\newtheorem{proposition}{Proposition}

\newcommand{\feed}[1]{\ensuremath{\mathsf{Parity^{(#1)}}}}
\newcommand{\fullfeed}[1]{\ensuremath{\mathsf{Full}^{(#1)}}}
\newcommand{\abfeed}[1]{\ensuremath{\mathsf{Universal}^{(#1)}}}

\newcommand{\gfeed}{\mathsf{Feed}\xspace}
\newcommand{\adv}{\mathsf{Adv}\xspace}

\newcommand{\bcc}{\mathsf{BCC}\xspace}

\renewcommand{\Pr}[1]{\mathbb{P}\left[\,#1\,\right]}
\newcommand\E[1]{\mathbb{E}\left[\,#1\,\right]}

\newcommand{\cQ}{\mathcal{Q}}

\newcommand{\ie}{{\it i.e.,}\xspace}
\newcommand{\eg}{{\it e.g.}\xspace}

\newcommand{\symdiff}{\;\triangle \;}
\newcommand{\gt}{$(n, k)$-Group-Testing\xspace}
\newcommand{\aadv}{$\alpha$-Malicious Adversary\xspace}
\newcommand{\xaa}{$\alpha$-Honest $x$-Avoiding Adversary\xspace}
\newcommand{\hadv}{$\alpha$-Honest Adversary\xspace}
\newcommand{\grouptesting}{Group Testing\xspace}
\newcommand{\dk}[1]{{#1}}
\newcommand{\dpa}[1]{{#1}}

\newcommand{\remove}[1]{}
\newcommand\bigexists{%
  \mathop{\lower0.75ex\hbox{\ensuremath{%
    \mathlarger{\mathlarger{\mathlarger{\mathlarger{\exists}}}}}}}%
  \limits}
\newcommand\bigforall{%
  \mathop{\lower0.75ex\hbox{\ensuremath{%
    \mathlarger{\mathlarger{\mathlarger{\mathlarger{\forall}}}}}}}%
  \limits}

\begin{document}

\title{Generalized Framework for Group Testing: \\ Queries, Feedbacks and Adversaries}

\author{
	Marek Klonowski\footnote{Wroclaw University of Science and Technology, Poland, E-mails: marek.klonowski@pwr.edu.pl,dominik.pajak@pwr.edu.pl}
	\and
	Dariusz R. Kowalski\footnote{Augusta University, Augusta, GA, USA, and SWPS University, Warsaw, Poland, E-mail: dkowalski@augusta.edu}
	\and
	Dominik Paj\k{a}k\footnotemark[1]}

\maketitle

\begin{abstract}%
In the \grouptesting problem, the objective is to 
learn a subset $K$ of some much larger domain $N$, using the shortest-possible sequence of queries $\mathcal{Q}$. A feedback to a query provides some information about the intersection between the query and 
 subset $K$. Several specific feedbacks have been studied in the literature, often proving different formulas for the estimate of the query complexity of the problem, defined as the shortest length of queries' sequence solving \grouptesting problem with specific feedback.
In this paper we 
study what are the properties of the feedback that influence the query complexity of \grouptesting and what is their measurable impact.
We propose a generic
framework that covers a vast majority of relevant settings considered in the literature, 
which depends on two fundamental parameters of the feedback: input capacity $\alpha$ and output expressiveness $\beta$. 
\dk{They upper bound the logarithm of the size of the feedback function domain and image, respectively.}
To justify the value of the framework, we prove upper bounds on query complexity of non-adaptive, deterministic \grouptesting under \dk{some ``efficient''} feedbacks, for minimum, maximum and general expressiveness,
and complement them with a lower bound on all feedbacks with given parameters $\alpha,\beta$.
\dk{Our upper bounds also hold if the feedback function could get an input twisted by a malicious adversary, in case the intersection of a query and the hidden set is bigger than the feedback capacity~$\alpha$.}
We also show that slight change in the feedback function may result in substantial worsening of the query complexity. 
\dk{Additionally, we analyze explicitly constructed randomized counterparts of the deterministic results.}
Our results provide some insights 
\dk{to}
what are the most useful bits of information an output-restricted feedback could provide,
and open a number of challenging research directions. 


\




 
\end{abstract}


\section{Introduction}

\grouptesting, introduced by 
\cite{dorfman1943detection}, is an inference problem, where the goal is to identify, by asking queries, all elements of an unknown set $K$. All we initially know about set $K$ is that $|K| \leq k$ and that it is a subset of some much larger set $N$ with $|N| = n$, for given parameters $k,n$. To learn set $K$, we must have answers to the queries that provide some information about set $K$. In our model, the answer to a query $Q$ depends on the intersection between $K$ and $Q$ and equals to $\gfeed(K \cap Q)$, where $\gfeed$ is some known pre-defined feedback function. The sequence of queries is a correct solution to \grouptesting if and only if for any  two different sets $K_1, K_2$ such that $|K_1|,|K_2| \le k$,
\dk{the sequence of feedback answers computed for sets $K_1$ and $K_2$ are different; we say then that the sequence of queries distinguishes any pair of sets, or identifies any set of size at most $k$.\footnote{%
\dk{In this work we abstract from 
computational efficiency
of decoding of sets,
which is a large research area by itself, c.f.,~\cite{Aldridge2014}.}}}
The objective is, 
for a given deterministic feedback function $\gfeed(\cdot)$,
to find 
a sequence of queries that identifies any set $K$ of size at most $k$ and the length of this sequence, called query complexity, will be shortest possible. In particular, we are interested in algorithms that have query complexity logarithmic in $n=|N|$ and polynomial in $k=|K|$. In 
the classic 
variant, 
studied in most of the existing relevant literature,
the function $\gfeed$ simply answers whether the intersection between $K$ and $Q$ is empty or not, \dk{while another popular feedback returns the size of the intersection~\cite{duhwang,gallager1985perspective}.} These variants were applied in many domains, including pattern matching~\cite{clifford2010pattern, Indyk97}, compressed sensing~\cite{cormode2006combinatorial}, streaming algorithms~\cite{cormode2005s} and graph reconstruction~\cite{choi2010optimal,GrebinskiK00}
 or even accelerating computations in neural networks~\cite{Liang2021}. Though one of the most prominent examples of applications of \grouptesting is in conflict resolution in communication networks~\cite{capetanakis1979generalized,capetanakis1979tree,chlebus2017randomized,gallager1985perspective,greenberg1987estimating,greenberg1985lower,KomlosG85,massey1981collision,wolf1985born}.

There is a large body of literature introducing new variants of the \grouptesting~\cite{Censor-HillelHL15,GrebinskiK00,Bshouty09,MarcoJKRS20,MarcoJK19}, which could be \dk{simply} viewed as different feedback functions applied to \dk{some generic} \grouptesting framework. Therefore, in this paper we aim at \dk{designing such a universal
framework allowing a holistic view at many previous modifications of \grouptesting setting,
and
study the dependence of the query complexity 
on two identified fundamental parameters of the feedback function:}
\begin{description}
\item[\emph{Capacity:}] this parameter denotes the maximum set size that can be processed by the feedback function. 
\dk{In other words, the domain of the feedback function with capacity $\alpha$ is the family of all subsets of $N$ of size at most $\alpha$.} 
The capacity is denoted by $\alpha$ throughout this paper and varies between $1$~and~$k$.
\item [\emph{Expressiveness:}] this parameter denotes the number of output bits of the feedback function. \\
It is denoted by $\beta$ throughout this paper, and varies between $1$ and $\bar{\alpha}$, where the latter denotes a binary logarithm of the number of all subsets of $N$ of size at most $\alpha$, \dk{i.e., $\bar{\alpha}=\log_2 \sum_{i=0}^{\alpha} \binom{n}{i}$.}
\end{description}

A feedback function with capacity $\alpha$ and expressiveness $\beta$ is called an {\em $(\alpha,\beta)$-feedback}.

%
If for some query $Q$, $|Q \cap K| > \alpha$, then the intersection set $Q\cap K$ cannot be passed directly to the feedback function, \dk{because $(\alpha,\beta)$-feedback functions are not defined for such sets.} 
Therefore, in our framework we resolve this issue by presence of an adversary \dk{-- a non-deterministic feature which,} for such queries with large intersection, 
selects a set with at most $\alpha$ elements
and the answer to query $Q$ is the feedback on this set. 
We consider different models of an adversary, \dk{including a powerful Malicious Adversary, who could ``fool'' the feedback with {\em arbitrary sets} of at most $\alpha$ elements of $N$,\footnote{%
\dk{One could also assume that the Malicious Adversary could give {\em any} input to the feedback function in case of exceeded capacity --
this could however require re-definition of the feedback function to handle 
non-valid input sets.}

\dk{Example of an adversary, widely but implicitly considered in the literature, is the mechanism of feedback in radio networks or multiple access channels, when the feedback provides answers ``collision'' or ``silence'' or an arbitrary element if two or more neighbors of a communication device transmit simultaneously, c.f.,~\cite{Bar-YehudaGI92,chlebus2017randomized}.}
}
and more benign Honest Adversary, who always returns {\em some subset} of the intersection.}

Clearly, increasing 
the capacity parameter $\alpha$ or expressiveness $\beta$ increases the number of $(\alpha,\beta)$-feedback functions,
and \dk{thus} should decrease the query complexity of the best feedbacks in this family. 
But what is the asymptotic \dk{pace of this query complexity decrease?} 
\dk{Is there a substantial difference in query complexity of \grouptesting under Honest and Malicious adversaries?}
Are there better and worse feedback functions for given $\alpha,\beta$, i.e., resulting in smaller (resp., larger) query complexity?
This paper provides partial answers to 
these questions.


\paragraph{Document structure}
In Section~\ref{sec:model} we formally define the general framework of \gt,
\dk{including
$(\alpha,\beta)$-feedback functions and adversaries,} and outline the contribution of the paper.
Then we discuss a related work on specific feedbacks in Section~\ref{sec:related-future}.
In Section~\ref{sec:upper} we prove upper bounds on the query complexity for efficient feedbacks with minimal, maximal and general expressiveness $\beta$ and any capacity $\alpha$, \dk{under powerful Malicious Adversary}. Section~\ref{sec:lower} presents a lower bound for feedbacks with maximum expressiveness, i.e., $(\alpha,\bar{\alpha})$-feedbacks, \dk{which holds even for more benign Honest adversaries.} A case study of two $(\alpha,2\log n)$-feedback functions with (provably) substantially different query performance is given in Section~\ref{sec:case-study}. 
Discussion of results from perspective of future directions is given in~Section~\ref{sec:future}. 
\section{Generalized framework and our contribution}
\label{sec:model}

\dk{As we will discuss in Sections~\ref{sec:technical-results} and~\ref{sec:related-future},
many previously considered variants of \gt problem could be expressed by, and their query
complexities depend on, specific parameters of the feedback to the queries. Here,}
we formally introduce \dk{generalized framework, including} families of $(\alpha,\beta)$-feedbacks, where $\alpha$ is the feedback capacity
while $\beta$ is its expressiveness, \dk{and adversaries that provide input to the feedback function in case the intersection has more than $\alpha$ elements.} 
\dk{We consider {\em non-adaptive deterministic} solutions, 
in which subsequent queries do not depend on the feedback from the previous ones nor on random bits. 
This class is very popular in the literature, due to its applicability and relevance to coding~\cite{kautz1964nonrandom} and information theory~\cite{erdos1963two,aldridge2019group}.}

\dk{In subsequent technical sections, we will be studying} query complexity of the whole classes of $(\alpha,\beta)$-feedbacks,
depending on parameters $n,k,\alpha,\beta$ \dk{and specific adversary}, as well as several interesting sub-classes.  
\dk{We will also discuss randomized counterparts of our deterministic solutions, c.f., Definition~\ref{def:randomized}, as in some cases they could be computed more~efficiently.}

%

We assume that the universe of all elements $N$, with $|N| = n$, is enumerated with integers $1,2,\dots,n$. Throughout the paper we will associate an element with its identifier. 

\paragraph{Specification of generalized \grouptesting framework}
\begin{definition}
\label{def:framework}
The generalized \gt framework is defined as follows:
\begin{enumerate}
\item An \gt Algorithm is defined as a sequence of queries $\cQ=\mathcal{Q}^{n,k} = \langle Q_1,Q_2,\dots,Q_t\rangle$ depending on $n$, $k$, where each query is an arbitrary subset of $N$.
\dk{The sequence length $t$ is called a query complexity of the sequence/algorithm.}\footnote{%
\dk{Due to the scope of this paper, our definition considers non-adaptive algorithms, i.e., in which the sequence of queries is fixed in advance.
However, an analogous framework can be defined for adaptive algorithm, in which consecutive queries are defined
based on the partial feedback vector, i.e., feedbacks on the preceding queries.}}

\item An adversary is defined as an entity that performs two actions. Firstly, it chooses set $K$ as an arbitrary subset of $N$ with $|K|\leq k$. Secondly, it defines a
 function $\adv(X,i \mid \mathcal{Q},K)$, for every $X \subseteq N$, and every $i \in \{1,\dots, |\mathcal{Q}|\}$, where $i$ denotes the index\footnote{\dk{This means that the adversary receives not only the whole sequence $\mathcal{Q}$ but also the step number; hence, may output different values for two identical intersection sets but obtained for different queues.}} in sequence $\mathcal{Q}$. 
 This function must satisfy: 
 \begin{itemize}
\item
$\adv(X,i \mid \mathcal{Q}, K))\subseteq N$,
\item
$|\adv(X,i \mid \mathcal{Q}, K))| \leq \alpha$,
\item
 $\adv(X,i  \mid \mathcal{Q}, K) = X$, if $|X| \leq \alpha$.
 \end{itemize}
\item An adversary strategy, \dk{under a given query sequence $\cQ$ and a set $K$ fixed by an adversary
} is defined as 
\dk{a}
function $\adv(X,i \mid \mathcal{Q}, K)$ of two arguments: set $X \subset N$ and index $i \in \{1,2,\dots, |\mathcal{Q}|\}$. 
\dk{$\mathcal{S}_{adv}(\cQ,K)$ denotes the set of all adversarial strategies under a given query sequence $\cQ$ and a set $K$ fixed by an adversary, and $\mathcal{S}_{adv}\dk{(\mathcal{Q},\cdot)=\{\mathcal{S}_{adv}(\cQ,K)\}_{K\subseteq N, |K|\le k}}$ is the set of all possible strategies of the adversary over sets $K$ of at most $k$ elements.}
\item An $(\alpha,\beta)$-feedback function $\gfeed$ is a function that takes as an input any subset of $N$ with at most $\alpha$ elements and 
\dk{outputs}
a binary vector of $\beta$ bits.
\item A feedback vector is defined as a sequence of outputs of the feedback function on the intersections between $K$ and the subsequent queries $Q_1,Q_2,\dots,Q_t$:
\begin{align*}
\mathcal{F}(K,\adv) = &\langle \gfeed(\adv(Q_1 \cap K, 1 \mid \mathcal{Q}, K)), \gfeed(\adv(Q_2\cap K, 2 \mid \mathcal{Q}, K)),\dots,\\& \gfeed(\adv(Q_t \cap K, t \mid \mathcal{Q}, K))\rangle \ ,
\end{align*}

\item For any fixed $n,k$, we say that a sequence of queries $\mathcal{Q}$ solves \gt problem under some adversary with the set of possible strategies $\mathcal{S}_{adv}\dk{(\mathcal{Q},\cdot)}$
if we have:
\begin{equation}\nonumber
 \bigforall_{\substack{K_1, K_2 \subset N\\ |K_1|, |K_2| \leq \alpha \\ K_1 \neq K_2}} \{\mathcal{F}(K_1,\adv) : \adv \in \mathcal{S}_{adv}\dk{(\cQ,K_1)}\} \cap  \{\mathcal{F}(K_2,\adv) : \adv \in \mathcal{S}_{adv}\dk{(\cQ,K_2)}\} = \emptyset 
 \end{equation}
\dk{In other words, the sets of possible (under the given adversary) feedback vectors for two different sets $K_1,K_2$ are disjoint.}
 \end{enumerate}
\end{definition}

\dk{In the above framework, the adversary could be deterministic (if the set of strategies $\mathcal{S}_{adv}(\cQ,K)$ for any given $\cQ,K$ is a single function, e.g., always passing an empty set to the feedback function for intersections larger than $\alpha$) or non-deterministic (otherwise). The feedback function is always deterministic. 
Observe also that in case of non-adaptive algorithms considered in this work the order of queries does not matter from perspective of query complexity, but 
helps in the analysis
to relate queries with their corresponding feedbacks in the feedback vector.}

\paragraph{Decoding of elements}
\dk{It follows from our Definition~\ref{def:framework}, point 6, of solving \gt problem that elements of the hidden set $K$ could be enlisted.
A straightforward, though not computationally efficient way, would be to consider all possible sets $K$ of size at most $k$; then, for each of them -- consider a family of all possible adversarial strategies and compute
feedback vectors for them; finally, one could find among them a matching copy of the actual feedback vector.
This copy is in some computed family corresponding to a set $K$, which is the actual hidden set to be enlisted.
The correctness of this solution follows directly from Definition~\ref{def:framework}, point 6: all possible feedback vectors obtained for all possible adversarial strategies are disjoint 
for different sets $K$ 
of size at most $k$.
In this work we do not study more efficient decoding algorithms than the above mentioned method -- this topic could be an interesting and challenging future direction.}


\paragraph{Maximum capacity}
The intersection between a query and set $K$ has always at most $k$ elements, hence having $\alpha$ larger than $k$ does not increase the power of the model compared to the case of $k = \alpha$. Therefore, in all our results we assume that $k \geq \alpha$ (for a setting $\alpha > k$ one could use a sequence of queries for $k = \alpha$).

\paragraph{Maximum expressiveness}
Similarly, we may restrict our considerations to $\beta \leq \bar{\alpha}$ because of the following fact.
\begin{proposition}
\label{fct:expres}
For any feedback function $f_1$, there exists a feedback function $f_2$ with expressiveness $\beta \leq \bar{\alpha}$ such that for any two sets $K_1, K_2 \subseteq N$, with $|K_1|, |K_2| \leq \alpha$, 
\[
f_1(K_1) = f_1(K_2) \Leftrightarrow f_2(K_1) = f_2(K_2) \ .
\]
\end{proposition}

\paragraph{Adversaries and feedback functions}
\label{results}


In this paper we consider the following adversaries and feedback functions. Note that one could consider also other types of adversaries and feedback~functions.
\begin{definition}
\label{def:adversaries}
We define the following two adversary 
\dk{types:}
\begin{enumerate}
\item \aadv. This adversary,  whenever for some query $Q$ we have $|Q\cap K| > \alpha$, choses an arbitrary subset of at most $\alpha$ elements from set $N$ and passes this set to the feedback function. Effectively, such adversary has the power to choose an arbitrary 
 value of feedback for queries that intersect with the hidden set $K$ on more than~$\alpha$~elements.
\item \hadv.  This adversary, whenever for some query $Q$ we have $|Q\cap K| > \alpha$, choses a subset of exactly $\alpha$ elements from set $Q\cap K$ and passes this set to the feedback function.  

\begin{itemize}
\item \xaa, is a special case of \hadv that for some element $x\in N$, if $x \in K \cup Q$ and $|K\cup Q| > \alpha$ then the set chosen by the adversary does \textbf{not} contain element $x$. In other words it hides element $x$, whenever possible. 
\end{itemize}

\end{enumerate}
\end{definition}

\begin{definition}
\label{def:feedbacks}
We define the following three feedback functions:
\begin{enumerate}
\item $\feed{\alpha}(X) = (|X|\text{ mod } 2)$. It is an $(\alpha,1)$-feedback, function because the returned value can be encoded on one bit.  
\item $\fullfeed{\alpha}(X) = X$. It is an $(\alpha,\bar{\alpha})$-feedback, as any subset of $N$ of at most $\alpha$ elements can be encoded by $\bar{\alpha}$~bits.
\item $\abfeed{\alpha,\beta}(X) = \left(|X| \text{ mod } 2\right) \bigparallel \left(\bigoplus_{x \in X} \mathsf{BCC}(x)\right),$ where $\mathsf{BCC}(x)$ is an  $\left[n,\beta- 1, \left\lfloor \frac{\beta - 1}{c \log \frac{n}{k}}\right\rfloor\right]$-BCC code of element $x$, \dk{c.f., Definition~\ref{def:BCC},} and $c$ is a constant from~\cite[Lemma 2]{Censor-HillelHL15}
$\bigoplus$ denotes bitwise XOR operation and $\bigparallel$ denotes concatenation of vectors. It is an $(\alpha,\beta)$-feedback, because $BCC$ code uses $\beta-1$ bits and the remaining bit denotes the parity of $|X|$.
\end{enumerate}
\end{definition}
%


\dk{The above Definition~\ref{def:adversaries} of adversaries and the third defined feedback function in Definition~\ref{def:feedbacks} have not been considered in the \grouptesting literature, to the best of our knowledge. 
We will derive upper bounds under the strongest of the defined adversaries, \aadv, while we also prove nearly matching lower bound(s) that holds also under the weaker \xaa; thus, the power of the adversary does not have a substantial impact on the query complexity of \grouptesting.

The next observation specifies useful criteria for the analysis of algorithms against \aadv, which we will apply in all our proofs of upper bounds.}

\begin{proposition}
\label{fct:technique}
Fix any $n,k$. If for query sequence $\mathcal{Q}^{n,k}= \left\langle Q_{i}\right\rangle_{i=1}^{t}$ we have that for any $K_1,K_2 \subset N$, with $|K_1|, |K_2| \leq \alpha$ and $K_1 \neq K_2$: 
\[
\bigexists_{\tau} |Q_{\tau} \cap K_1| \leq \alpha  \wedge  |Q_{\tau} \cap K_2| \leq \alpha \wedge \gfeed(Q_{\tau} \cap K_1) \neq \gfeed(Q_{\tau} \cap K_2) \ ,
\]
then $\mathcal{Q}^{n,k}$ solves \gt under \aadv.
\end{proposition}
In the following we will say that a query $Q_{\tau}$ \emph{distinguishes} sets $K_1$  and $K_2$ under some feedback function $\gfeed$ if $|Q_{\tau}\cap K_1| \leq \alpha$, $|Q_{\tau}\cap K_2|\leq \alpha$ and $\gfeed(Q_{\tau} \cap K_1) \neq \gfeed(Q \cap K_2)$.
\subsection{Technical results}
\label{sec:technical-results}
\paragraph{Binary feedback} First, we consider feedbacks with minimum possible expressiveness, namely, 
returning
only one bit of information. In this setting we have to answer the question of \emph{What is the most useful bit of information about a set of elements?} It turns out that a parity bit allows us to obtain an efficient solution in the family of $(\alpha,1)$-feedbacks. Interestingly, this result, and all our other upper bounds, hold for the strongest adversary.
\begin{theorem}
\label{thm:binary}
Under $\feed{\alpha}$ feedback and under \aadv, there exists a \dk{deterministic} solution to \gt with query complexity $O\left(\left(k +\frac{k^2}{\alpha}\right) \cdot \log \frac{n}{k}\right)$.
\end{theorem}
The proof is based on derandomization of random queries drawn from different random distribution, after proving
that these queries satisfy a certain Separation Property (formulated and proved in Lemma~\ref{lem:telescope}).

\paragraph{Full feedback} 
Our second result is in the setting with
the 
maximum possible expressiveness $\beta=\bar{\alpha}=\Theta(\alpha\log(n/\alpha))$, i.e., sufficient to return all identifiers of any set of size at most $\alpha$. 
We show that maximum expressiveness allows to design algorithms with small query complexity $O\left(\min\left\{\frac{n}{\alpha},\frac{k^2}{\alpha^2} \cdot \log^c n\right\}\right)$ for some $c\in [1,2]$, \dk{more precisely:} 
\begin{theorem}
\label{thm:fullupper}
Under $\fullfeed{\alpha}$ feedback and under \aadv, 
there exists a \dk{deterministic} solution to \gt with query complexity: 
\begin{align}\nonumber
& O\left(\min\left\{\frac{n}{\alpha},\frac{k^2}{\alpha^2} \cdot \log n\right\}\right) & \text{if } \alpha > 18 \log k, \\\nonumber
& O\left(\min\left\{\frac{n}{\alpha},\frac{k^2}{\alpha} \cdot \log\frac{n}{k}\right\}\right) & \text{otherwise}.
\end{align}
\end{theorem}
The proof is via derandomization of a random sequence of queries $\mathcal{Q}$, from which we require 
to simultaneously satisfy two conditions: on the number of queries containing a specific element, and on the sizes of the intersections of queries from any subset of $\mathcal{Q}$ of certain size and any possible instantiation of set $K$.

Interestingly, for $\alpha =\omega(\sqrt{k\log n})$, the obtained query complexity is sublinear in $k$. This can be contrasted with an $\Omega(k \log(n/k))$ lower bound for classical \grouptesting{} \dk{(i.e., for $\alpha=O(1)$)} that holds also for any randomized algorithm working with non-vanishing probability~\cite{coja2020information}. \dk{This proves the impact of feedback capacity on query complexity.}

\paragraph{General feedback} 
After considering both extreme values of $\beta$ we study the general case, where a feedback needs to work for an arbitrary $1\le \beta\le \bar{\alpha}$. In this case our first contribution is a design of  a more sophisticated general feedback function $\abfeed{\alpha,\beta}$, 
c.f., Definition~\ref{def:feedbacks}, which works for almost any $\alpha, \beta$. Our proposed feedback is a concatenation of a specific code (called BCC code) with an additional parity bit. Under this feedback we obtain the main result of the paper:
\begin{theorem}
\label{thm:generalupper}
Under $\abfeed{\alpha,\beta}$ feedback and under 
\dk{\aadv,}
there exists a \dk{deterministic} solution to \gt with query complexity $O\left(\frac{k^2}{\alpha\beta}\log^{c+1} n\right)$ for some $c\in [1,2]$,
\dk{more precisely:} 
\begin{align}\nonumber
& O\left(\frac{k^2}{\alpha\beta}\log n\left(\frac{\beta}{\alpha} + \log n\right)\right) & \text{if } \alpha > 18 \log k, \\\nonumber
& O\left(\frac{k^2}{\alpha} \cdot \log\frac{n}{k}\right) & \text{otherwise}.
\end{align}

\end{theorem}
Our main result shows that the query complexity decreases linearly with $\alpha$ and with $\beta$. Intuitively factor $\frac{k}{\alpha}$ in our complexity comes from \emph{congestion}, since the feedback function has capacity to serve at most $\alpha$ elements out of $k$ in a single query. The second factor $\frac{k \log\frac{n}{k}}{\beta} \approx \frac{\log {n \choose k}}{\beta}$ comes from the information-theoretic bound that we need $\log_2 {n \choose k}$ bits to uniquely encode any subset of $k$ elements and the fact that the feedback function provides only $\beta$ bits per round. 
What is surprising and challenging to prove is that the query complexity of efficient (but not all!) $(\alpha,\beta)$-feedbacks is (close to) a multiplication of these two characteristics.

The proof combines ideas from the analysis of the binary feedback and full feedback. In the binary feedback case we observe that sets that differ on many elements can be distinguished quickly using the parity feedback. On the other hand, sets that differ only on few elements are handled using a combination of full feedback algorithm with a specific coding to encapsulate the feedback into $\beta$ bits.

\paragraph{Lower bound}
We show a lower bound that proves that our upper bound shown in Theorem~\ref{thm:fullupper} is optimal up to polylogarithmic factor, for any $\alpha$.
\dk{It holds} even for a weaker adversary, \dk{\hadv, or more specifically, for its sub-type of \xaa. Thus, it also holds for the stronger \aadv, for which all our algorithms are~analyzed.}


\begin{theorem}
\label{thm:fulllower}
If $n > k^2\log n/\log k$, then any \dk{deterministic} solution to \gt under any $(\alpha,\beta)$-feedback has query complexity 
$\Omega\left(\frac{k^2}{\alpha^2} \log^{-1} k\right)$ for some \hadv.
\end{theorem}

\dk{
The proof of Theorem~\ref{thm:fulllower} is by transformation of our generalized \grouptesting framework to selectors -- structures studied in related literature, formally defined in Section~\ref{sec:lower}. We show that if there were shorter query sequences, there would exist selectors violating some of their lower bound. This transformation is however possible only in one way, as we will show in the next result.}

\paragraph{Minimum Elements feedbacks} Our two final results show that designing an efficient feedback function is very subtle. We show that a reasonable $(\alpha,2\log n)$-feedback function that returns two minimal elements from the set leads to very large query complexity of $\Omega(\min\{n,k^2\})$ if we restrict the function to return the elements \emph{in fixed order}, c.f., Theorem~\ref{thm:lower-2min}. Without this restriction it is possible to obtain feedback function for which there exists a \dk{deterministic} algorithm with query complexity $O\left(\frac{k^2}{\alpha} \cdot \log \frac{n}{k}\right)$, 
c.f., Corollary~\ref{cor:upper-2min}.

\dk{Theorem~\ref{thm:lower-2min} with Corollary~\ref{cor:upper-2min} provide an argument that there is no universal reduction between  selectors and our general \grouptesting framework, 
as both the considered feedback functions have the same parameters $\alpha$ and $\beta$ and differ only (slightly) in the definition of the feedback function, yet having query complexities different nearly by factor $\alpha$. Thus, our general framework is provably more complex than the theory of selectors.}

\begin{table}
	\centering
	\begin{tabular}{llll}
		\toprule
		$\alpha$ & $\beta$ & Upper bound & Lower bound \\\midrule
		$1$ & $1$ & $O\left(k^2 \log\frac{n}{k}\right)$~\cite{BonisGV03} & $\Omega\left(k^2 \frac{\log n}{\log k} \right)$~\cite{ClementiMS01} \\\midrule
		$k$ & $1$ & $O(k \log \frac{n}{k})$~\cite{Censor-HillelHL15} & $\Omega(k \log \frac{n}{k})$~\cite{Censor-HillelHL15}  \\\midrule
		$k$ & $\log k$ & $O\left(k \frac{\log \frac{n}{k}}{\log k}\right)$~\cite{GrebinskiK00} &$\Omega\left(k \frac{\log \frac{n}{k}}{\log k}\right)$ ~\cite{djackov1975search, lindstrom1975determining}\\\midrule
		$*$  & $1$ & $O\left(\left(k + \frac{k^2}{\alpha}\right) \log\frac{n}{k} \right)$ Thm~\ref{thm:binary} &$\Omega\left(\frac{k^2}{\alpha^2} \log^{-1}k \right)$ Thm~\ref{thm:fulllower} \\\midrule
	$*$  & $\bar{\alpha}$ & $O\left(\min\left\{\frac{n}{\alpha},\frac{k^2}{\alpha^2} \log n  \right\}\right)$ Thm~\ref{thm:fullupper} & $\Omega\left(\frac{k^2}{\alpha^2} \log^{-1}k \right)$ Thm~\ref{thm:fulllower} \\\midrule
		$*$  & $*$ & $O\left(\frac{k^2}{\alpha\beta} \log^2 n \right)$ Thm~\ref{thm:generalupper} & $\Omega\left(\frac{k^2}{\alpha^2} \log^{-1}k \right)$ Thm~\ref{thm:fulllower}\\		\bottomrule
	\end{tabular}
	\caption{\label{tab1} Results
	on non-adaptive \gt with $(\alpha,\beta)$-feedback. The upper bound column states query complexity of the best found $(\alpha,\beta)$-feedback found for 	parameters $\alpha,\beta$ fixed in the first two columns; as we will show, not all $(\alpha,\beta)$-feedbacks could reach that complexity. Symbol $*$ stands for any valid 
value of the parameter, and $\bar{\alpha}$ stands for a ceiling of the binary logarithm of the number of all subsets of $N$ of size at most $\alpha$. We display results from Theorem~\ref{thm:fullupper} and~\ref{thm:generalupper} in regime $\alpha > 18 \log k$, however our theorems cover the whole range of $\alpha$. }
	\end{table}

Table~\ref{tab1} presents our main \dk{deterministic} results in comparison to the most related previous work on specific feedback~functions.
\dk{In this work we also analyze explicitly constructed randomized counterparts of the deterministic results.}

\dpa{
\begin{definition}
\label{def:randomized}
A randomized algorithm solves \gt against Adaptive Adversary with probability $1-c$, for some $0\leq  c < 1$, if with probability $1-c$ it generates a sequence of queries $\mathcal{Q}$ that solves \gt according to Defintion~\ref{def:framework}.
\end{definition}
Note that in Definition~\ref{def:randomized} the adversary is assumed to know sequence $\mathcal{Q}$ (see Defintion~\ref{def:framework}(3)). Hence, 
\dk{our analysis' of randomized counterparts of deterministic solutions also hold}
against \emph{Adaptive Adversary}. This is to distinguish from the case, where the adversary does not know all the queries when choosing set $K$~\cite{BonisV17, bay2020optimal}.}




\section{Motivation, previous and related work}
\label{sec:related-future}

The problem of \grouptesting (and related equivalent problems such as coin weighting) has been considered in various feedback models. In this section we present details of implementation of some classical feedback models in our framework. Our framework, with two parameters of feedback $\alpha$ and~$\beta$, allows, among others, a comparison of results in different models, for a discussion about what is the best utilization of feedback output bits, and for comparison and generalization of existing results obtained for specific feedbacks, \dk{c.f., Table~\ref{tab1}.}

\paragraph{Beeping model and shared channel communication}
Beeping feedback model is a standard model considered in most of the \grouptesting literature~\cite{duhwang}, where the feedback tells whether the intersection between query $Q$ and set $K$ is empty or not.  Solutions to \grouptesting in this feedback model have direct applications to conflict resolution on a multiple access channel and broadcast in unknown radio networks, c.f.,~\cite{ClementiMS01}.

Observe that in 
 Beeping
feedback model, the feedback returns $0$ if the intersection is empty and $1$ otherwise. Thus beeping feedback is a $(1,1)$-feedback. 

In this feedback model, the \grouptesting problem is known to be solvable using $O(k^2\log(n/k))$~\cite{BonisGV03} queries and an explicit construction of length $O(k^2 \log^2 n)$~\cite{kautz1964nonrandom} exists. Best known lower bound (for $k < \sqrt{n}$) is $\Omega(k^2 \log n/ \log k )$~\cite{ClementiMS01}. 

A related model, where the feedback equals NULL if the intersection is of size $0$, the identifier of the element, if the intersection is of size $1$ and a value COLLISION otherwise, can be seen as $(2,\log n)$-feedback. This model is applicable to communication on shared channel and has been an area of extensive research. The solutions in literature include adaptive algorithms~\cite{capetanakis1979generalized, capetanakis1979tree}, semi-oblivious algorithms where an element can deactivate after successful transmission~\cite{KomlosG85, greenberg1985lower} (see surveys~\cite{gallager1985perspective,chlebus2017randomized} for more details on results in this model).
\dk{As mentioned earlier, some of the previous works also consider non-adaptive adversarial component of the feedback, c.f.,~\cite{Bar-YehudaGI92}.}

\paragraph{Finite-field additive radio network}
In this model, the feedback to a query is a parity of the size of the intersection between set $K$ and a query. 
One can observe that using $BCC$-codes of length $O(k \log \frac{n}{k})$~\cite{Censor-HillelHL15} it is possible to design a sequence of 
 queries of the same length, that solves \gt in this model. 
This construction solves \gt with $O(k \log \frac{n}{k})$ in a $\feed{k}$ feedback model (which is an example of $(k,1)$-feedback), because by the definition of $BCC$-codes any bit-wise XOR of up to $k$ codewords is unique.  The construction of BCC codes has also been applied to solutions of standard communication problems (such as broadcast) in specific models of communication networks~\cite{Censor-HillelHL15}.

 We note that $\feed{\alpha}$ feedback for borderline value of $\alpha = k$ corresponds to the setting considered in~\cite{Censor-HillelHL15}. In this case our algorithm matches the best known upper bound, hence our proposed feedback function and our algorithm are a valid generalization, showing the smooth transition of query complexity
between settings of $\alpha = \log k$ and $\alpha = k$ in a pace inversely proportional to the feedback capacity $\alpha$.

 \paragraph{Coin weighting}
 The problem of coin weighting is exactly the \grouptesting problem with a different feedback. In the coin weighting problem, we have a set of $n$ coins of two distinct weights $w_0$ (true coin) and $w_1$ (counterfeit coin), out of which up to $k$ are counterfeit ones. We are allowed to weight any subset of coins in a spring scale, hence we can deduce the number of counterfeit coins in each weighting. The task is to identify all the counterfeit coins. 
 
 The coin weighting can be implemented in our framework as a $(k,\log k)$-feedback, where the feedback returns the size of the intersection between the query and the set $K$. The problem is solvable with $O(k\log (n/k)/ \log k)$~\cite{GrebinskiK00} queries. 
 
Bounds for both $BCC$ codes and non-adaptive coin weighting are tight, thus increasing the number of output bits from $1$ to $\log k$ results in decrease in query complexity by a factor of $\log k$.

\paragraph{Threshold \grouptesting}
In this variant of \grouptesting introduced in~\cite{DAMASCHKE}, a number of thresholds $0 < t_1 \leq t_2 \leq \dots \leq t_s$ are defined. Thresholds divide the set $[k]$ into set of discrete intervals $[0,t_1), [t_1,t_2),\dots, [t_{s-1},t_{s}),[t_s,k]$. The feedback to query $Q$ is the index of the interval to which $|K \cap Q|$ belongs. This feedback can be implemented as a $(t_s + 1, \log s)$-feedback. An upper bound for a single threshold $t$ of approximately $O(\frac{k^2}{\sqrt{t}} \log\frac{n}{k})$~\cite{MarcoJKRS20} suggests that single threshold feedback is probably not the optimal feedback (according to our parameters) since we know that $(t,1)$-feedbacks can lead to query complexity $O(\frac{k^2}{t} \log\frac{n}{k})$. On the other hand, in~\cite{MarcoJK19} the authors analyze a feedback with $\sqrt{k \log k}$ thresholds out of which maximum threshold is $\Omega(k)$, which in our framework translates to a $(k, \log k)$-feedback. Result in~\cite{MarcoJK19} is an algorithm with query complexity $O(\frac{k}{\log k} \cdot \log \frac{n}{k})$, which is \dk{
logarithmically far from $O\left(\frac{k}{\log k} \log^2 n\right)$ obtained from our generic upper bound $O\left(\frac{k^2}{\alpha \beta} \log^2 n\right)$
in Theorem~\ref{thm:generalupper} instantiated for $\alpha = k$, $\beta = \log k$.} 

\paragraph{Other related results}
The problem of \grouptesting has been recently discussed from different perspectives. Some papers consider  different models of generating (or constraining) the subset $K$. This may lead to critically different optimal strategies, even for non-adaptive settings.  In~\cite{Aldridge2019} the author considers the model, wherein each element is included in $K$ with a fixed probability $p$ -- we need $\Omega(n)$ tests to have error probability tending to zero. Somehow related randomized model has been discussed in~\cite{BonisV17}, wherein the algorithm may fail on a small fraction of inputs. 
In \cite{Inan2020} the authors consider ``sparse'' \grouptesting, where the size of each query is limited. They also consider settings wherein each element can be included in a limited number of queries.

\section{Upper bounds}\label{sec:upper}

\subsection{Minimal expressiveness -- Binary feedback}
\label{sec:binary}

We first show (Lemma~\ref{lem:telescope}) an upper bound on length of a sequence that distinguishes any pair of sets satisfying a certain size restriction. 
This length is inversely proportional to the product of capacity $\alpha$ and the lower bound on the size of the symmetric difference between the sets, denoted by $\delta$.
This proof is based on analyzing a certain Separation Property of a sequence of random queries drawn from specific
probabilistic distribution, and showing that it yields distinguishing between two sets $K_1,K_2$ with a large probability, sufficient
to derandomize it.
In the second step (Lemma~\ref{lem:binarydelta}), we show how to remove the size restrictions from the result.
Finally, Theorem~\ref{thm:binary} will follow directly from Lemma~\ref{lem:binarydelta} applied for $\delta=1$.

In Lemma~\ref{lem:binarydelta} we will need the following notation and basic facts.
\paragraph{Basic notation and tools}
We will use the following notation for the symmetric difference 
of 
two sets $A \symdiff B = (A \setminus B) \cup (B \setminus A)$. In our proofs, we also use the following two elementary~facts:
\begin{fact}
\label{clm:oddeven}
Let $X \sim \mathsf{Binomial}(n,p)$, then $\Pr{X \text{ is odd}} = \frac{1}{2} - \frac{1}{2} (1-2p)^n$.
\end{fact}
\dk{The proof of Fact~\ref{clm:oddeven} can be found in the appendix.} The following fact can be found \eg, in \cite[(p. 34 eq. 6)]{mitrinovic1970analytic}.
\begin{fact}
\label{clm:bern_rev}
For any $0 \leq x \leq 1$ and $n \in \mathbf{N}_+$:
$
(1-x)^n \leq 1- nx + \frac12 n(n-1)x^2.
$
\end{fact}

\paragraph{\dk{Main technical tools}}
\dk{We first show how to construct sequences distinguishing pairs of sets $K_1,K_2$ satisfying specific conditions.}

\begin{lemma}
\label{lem:telescope}
For any $1 \leq \delta \leq k/\alpha$ and if $k \geq \alpha$, there exists a sequence of  $O\left(\frac{k^2}{\alpha \delta} \cdot \log (n/k)\right)$ sets $\mathcal{Q}$ such that for any two sets $K_1,K_2\subseteq N$ satisfying $k \geq |K_1| \geq k/2$ and $|K_1| \geq |K_2|$ and $|K_1 \symdiff K_2| \geq \delta$ there exists $Q \in \mathcal{Q}$ that satisfies $|Q\cap K_1| \leq \alpha$, $|Q\cap K_2|\leq \alpha$ and $\feed{\alpha}(Q \cap K_1) \neq \feed{\alpha}(Q \cap K_1)$.

\end{lemma}
\begin{proof}
We will show this result using the probabilistic method. 
More precisely, we first define a sequence of random queries $\mathcal{Q}$ of length $O(\frac{k^2}{\alpha\delta} \cdot \log (n/k))$.
Next, we fix any two different sets $K_1,K_2 \subseteq N$ whose cardinalities satisfy the conditions of the lemma. Recall that, \dk{by Proposition~\ref{fct:technique},} a query $Q \in \mathcal{Q}$ \emph{distinguishes} $K_1$ from $K_2$ if it satisfies the three conditions from the statement of the Lemma: $|Q\cap K_1| \leq \alpha$, $|Q\cap K_2|\leq \alpha$ and $\feed{\alpha}(Q \cap K_1) \neq \feed{\alpha}(Q \cap K_1)$. We will compute the probability that no query from sequence $\mathcal{Q}$ distinguishes the considered sets $K_1$ and $K_2$.
Then, we apply the union bound over all pairs of $K_1,K_2$ and take the complementary event, which, as we show, holds with a positive probability. This implies existence of the sought query sequence. The details follow.

\noindent
{\em Definition of random sequence $\mathcal{Q}$.}
We define a sequence of probabilities $\mathcal{P}$ of length $O(\frac{k^2}{\alpha\delta} \cdot \log (n/k))$ as probability $ \frac{\alpha}{16k}$ repeated $\left\lceil \frac{150 k^2}{\alpha\delta} \cdot \log \frac{4n}{k} \right\rceil$ times.
We define $Q_i$, an $i$-th element of the sequence $\mathcal{Q}$, as a set generated by including each element of $N$ independently with the $i$-th probability from sequence $\mathcal{P}$. 

\noindent
{\em Proving Separation Property.} 
Consider two different sets $K_1,K_2 \subseteq N$ whose cardinalities satisfy the conditions of the lemma.
Let $S = K_1 \symdiff K_2$ be the symmetric difference of $K_1$ and $K_2$. Note that $2k \geq |K_1 \cup K_2| \geq |S| \geq \delta$. We also know, by the assumed restriction on the size of $K_1$, that $|K_1 \cup K_2| \geq k/2$. 
We want to show the following:
\begin{quote}
{\bf\em  Separation Property:} for any positive integer $i\leq |\mathcal{Q}|/2$ and for some constant $c > 0$, the probability that query $Q_i \in \mathcal{Q}$,
distinguishes $K_1$ and $K_2$ is at least $c\alpha\delta/k$.
\end{quote}
Before proving the Separation Property we need the following technical claim.

\noindent
{\bf Claim.} For any $Q_j\in \mathcal{Q}$, where $j\leq |\mathcal{Q}|$:
\begin{align*}
\Pr{ Q_j \text{ distinguishes } K_1 \text{ from } K_2} & \geq \Pr{|(K_1 \cup K_2) \cap Q_j| \leq \alpha \text{ and } |S \cap Q_j| \text{ is odd}} \ .
\end{align*}

\noindent
{\em Proof of the Claim.}
%
Consider a query $Q_j$, for some $j\leq |\mathcal{Q}|$, from the random sequence $\mathcal{Q}$, and assume that the event ``$|(K_1 \cup K_2) \cap Q_j| \leq \alpha \text{ and } |S \cap Q_j| \text{ is odd}$'' holds. It follows from $|(K_1\cup K_2) \cap Q_j| \leq \alpha$ that $|K_1 \cap Q_j| \leq \alpha$ and $|K_2 \cap Q_j| \leq \alpha$. Moreover, since the event also implies that $|S \cap Q_j|$ is odd, then $|K_1 \cap Q_j| \neq |K_2 \cap Q_j| \mod 2$. Hence, $\feed{\alpha}(K_1 \cap Q_j) \neq \feed{\alpha}(K_2 \cap Q_j)$.
This completes the proof of the Claim. \ $\blacksquare$



We continue the proof of the Separation Property. 
In the case, where $\alpha \geq \log\frac{256k}{7\alpha \delta}$ we have, that:
\[ 
\Pr{|(K_1 \cup K_2) \cap Q| \leq \alpha \text{ and } |S \cap Q| \text{ is odd}}  \geq \Pr{|S \cap Q| \text{ is odd}}  - \Pr{|(K_1 \cup K_2) \cap Q| > \alpha}
\ .
\]
Random variable $H_2 = |(K_1 \cup K_2) \cap Q|$ is distributed according to the Binomial distribution with parameters $|K_1 \cup K_2|$ and $p$, and $\E{H_2} = p |K_1 \cup K_2| \leq 2pk = \alpha /8$. Then, by Chernoff bound (c.f.,~\cite{SURV}):
\[
\Pr{H_2 \geq  \alpha} \leq 2^{-\alpha} \ .
\]
Random variable $|S \cap Q|$ is also distributed according to the Binomial distribution with parameters $|S|$ and~$p$. By Fact~\ref{clm:oddeven} we have
\[
\Pr{|S \cap Q| \text{ is odd}} = \frac{1}{2} - \frac{1}{2}(1-2p)^{|S|} \ .
\]
Term $\frac{1}{2}(1-2p)^{|S|}$ is maximized, when $|S|$ is minimized, which gives us:
\[
\Pr{|S \cap Q| \text{ is odd}} \geq \frac{1}{2} - \frac{1}{2}(1-2p)^{\delta} \ .
\]
Using Fact~\ref{clm:bern_rev}, we get:
\begin{align*}
\frac{1}{2} - \frac{1}{2}(1-2p)^{\delta} \geq \frac12 - \frac12\left(1- 2p\delta + 4\delta(\delta-1)p^2\right) = \frac{\alpha \delta}{16k} - \frac{\delta(\delta-1) \alpha^2}{128 k} = \frac{\alpha \delta}{16k}\left(1 - \frac{(\delta - 1)\alpha}{8k}\right) \ ,
\end{align*}
and knowing that $\delta < k/\alpha$ we get $\left(1 - \frac{(\delta - 1)\alpha}{8k}\right) \geq 7/8$. Finally, combining the above and knowing that $\alpha \geq \log\frac{256k}{7\alpha \delta}$,
we get:
\[
 \Pr{|S \cap Q| \text{ is odd}}  - \Pr{|(K_1 \cup K_2) \cap Q| > \alpha}
\geq \frac{7\alpha\delta}{128k}  -  2^{-\alpha} \geq 
\frac{7\alpha \delta}{256k} \geq \frac{\alpha \delta}{50k} \ .
\]

In the second case assume, that $\alpha \leq \log\frac{256k}{7\alpha \delta}$. In this case we have:
\begin{align*}
\Pr{ Q_j \text{ distinguishes } K_1 \text{ from } K_2} & \geq \Pr{|(K_1 \cup K_2) \cap Q_j| \leq \alpha \text{ and } |S \cap Q_j| \text{ is odd}} \\
&  \geq \Pr{|(K_1 \cap K_2) \cap Q_j| \leq \alpha - 1 \text{ and } |S \cap Q_j| = 1} \\
& \geq \Pr{|(K_1 \cap K_2) \cap Q_j| \leq \alpha - 1 } \cdot \Pr{|S \cap Q_j| = 1},
\end{align*}
where the last equality is true, because sets $S$ and $K_1\cap K_2$ are disjoint. 

Random variable $H_1 = |(K_1 \cap K_2) \cap Q|$ is distributed according to the Binomial distribution with parameters $|K_1 \cap K_2|$ and $p$, and $\E{H_1} = p |K_1 \cap K_2| \leq pk = \alpha /16$. Then, by Markov inequality $\Pr{H_2 \geq  \alpha} \leq 1/16$.
We want to lowerbound term 
\[
\Pr{|S \cap Q_j| = 1} = p  |S| (1-p)^{|S|-1} = \frac{\alpha|S|}{16k}\cdot \left(1-\frac{\alpha}{16k}\right)^{|S|-1}
\]
If $|S| \leq 16k / \alpha$, then $(1-\alpha/(16k))^{|S|-1} \geq e^{-1}$ and:
\[
\Pr{|S \cap Q_j| = 1} \geq \frac{\alpha\delta}{16 e k}.
\] 
Hence $\Pr{ Q_j \text{ distinguishes } K_1 \text{ from } K_2} \geq \frac{15}{16} \cdot \frac{\alpha\delta}{16 e k} \geq \frac{\alpha\delta}{50 k}$.

If $|S| > 16 k / \alpha$, then $p|S| \geq 1$ and  (knowing that $|S| \leq 2k$), we get:
 \[
\left(1-\frac{\alpha}{16k}\right)^{|S|-1} \geq \left(1-\frac{\alpha}{16k}\right)^{\left(\frac{16k}{\alpha} -1\right) \frac{\alpha}{8} + \frac{\alpha}{8}} \geq e^{-\alpha / 8} \cdot e^{-\alpha / 8} \geq e^{-\alpha} \geq  \frac{7\alpha \delta}{256k}.
\]
Hence, also in this case, we get $\Pr{ Q_j \text{ distinguishes } K_1 \text{ from } K_2} \geq \frac{15}{16} \cdot \frac{7\alpha \delta}{256k} \geq \frac{\alpha\delta}{50 k}$.

This completes the proof of the Separation Property for $c=1/50$.

\noindent
{\em Computing the probability of $\mathcal{Q}$ distinguishing $K_1$ from $K_2$.}
By the proven Separation Property for $c=1/50$ and by independence of selection of each query in the sequence $\mathcal{Q}$ of length $150k^2/(\alpha\delta)\cdot \log(4n/k)$, the probability that 
$\mathcal{Q}$ fails to distinguish $K_1$ from $K_2$ is
at most:
\[
\left(1-\frac{\alpha\delta}{50 k}\right)^{150k^2/(\alpha\delta)\cdot \log(4n/k)} 
=
\left(1-\frac{\alpha\delta}{50 k}\right)^{50k/(\alpha\delta)\cdot 3k\log(4n/k)} 
\leq e^{-3k\log(4n/k)} = \left(\frac{4n}{k}\right)^{-3k} \ .
\]
{\em Applying the union bound and probabilistic argument.}
The number of possible pairs of sets $K_1,K_2$  for the case  $k \leq n/2$ can be upper bounded as follows: 
\[
\left(\sum_{i=1}^k {n\choose i}\right)^2 \leq k^2 {n\choose k}^2 \leq k^2\left(\frac{en}{k}\right)^{2k} \ .
\]
If $n \geq k \geq n/2$ the number of possible pairs $K_1,K_2$ can be simply upper bounded by:
\[
2^n \cdot 2^n \leq \left(\frac{4n}{k}\right)^{2k} \ .
\]
In both cases the number of possible pairs of $K_1,K_2$ is upper bounded by 
 $$ k^2 \cdot \left(\frac{4n}{k}\right)^{2k} \ .$$
Thus, using the Union Bound, the probability that some pair of sets is not distinguished by $\mathcal{Q}$ is at most:
\[
\left(\frac{4n}{k}\right)^{-3k} \cdot k^2 \cdot \left(\frac{4n}{k}\right)^{2k} \leq 4^{-k} \cdot k^2  < 1 \ .
\]
Hence, there is a positive probability of the complementary event that there exists a sequence of length $O(\frac{k^2}{\alpha\delta} \cdot \log (n/k))$ that distinguishes any pair $K_1$ and $K_2$ (satisfying the conditions of the lemma) under the $\feed{\alpha}$ feedback function, and by the probabilistic argument -- such a sequence exists.
\end{proof}

\dpa{
In the next lemma we show that it is possible to also distinguish sets of size at most $\alpha$.
\begin{lemma}
\label{lem:telescope2}
If $k \leq \alpha$, there exists a sequence of  $O\left(k \cdot \log (n/k)\right)$ sets $\mathcal{Q}$ such that for any two sets $K_1,K_2\subseteq N$ satisfying $k \geq |K_1|$, $k \geq |K_2|$  and $K_1 \neq K_2$ there exists $Q \in \mathcal{Q}$ that satisfies $\feed{\alpha}(Q \cap K_1) \neq \feed{\alpha}(Q \cap K_1)$.
\end{lemma}
\begin{proof}
The proof follows in a similar vein as proof of Lemma~\ref{lem:telescope}.
We construct a sequence of  $\lceil 3k \log (4n/k) \rceil$ queries $\mathcal{Q}$ as follows: $i$-th element of the sequence is generated by including each element of $N$ independently with probability $\frac12$. Since $|K_1| \leq \alpha$ and $|K_2| \leq \alpha$, then 
$\Pr{ Q_j \text{ distinguishes } K_1 \text{ from } K_2} \geq \Pr{|(K_1 \symdiff K_2) \cap Q_j| \text{ is odd}}$. By Fact~\ref{clm:oddeven} we have $\Pr{|(K_1 \symdiff K_2) \cap Q_j| \text{ is odd}} = \frac12$.

Similarly as in Lemma~\ref{lem:telescope}, the number of possible pairs of $K_1,K_2$ can be upper bounded by $ k^2 \cdot \left(\frac{4n}{k}\right)^{2k}$.
Thus, using the Union Bound, the probability that some pair of sets is not distinguished by $\mathcal{Q}$ is at most:
\[
\left(\frac{4n}{k}\right)^{-3k} \cdot k^2 \cdot \left(\frac{4n}{k}\right)^{2k} \leq 4^{-k} \cdot k^2  < 1 \ .
\]
Hence, there is a positive probability of the complementary event that there exists a sequence of length $O(k \cdot \log (n/k))$ that distinguishes any pair $K_1$ and $K_2$ (satisfying the conditions of the lemma) under the $\feed{\alpha}$ feedback function, and by the probabilistic argument -- such a sequence exists.
\end{proof}
}

In the next lemma we show that the sequences constructed in Lemma~\ref{lem:telescope} and Lemma~\ref{lem:telescope2} could be concatenated in order to obtain a sequence that distinguishes sets without the lower restriction on their sizes.

\begin{lemma}
\label{lem:binarydelta}
There exists a sequence $\mathcal{Q}$  of length  $O\left(\left(k+\frac{k^2}{\alpha \delta}\right) \cdot \log (n/k)\right)$ for any $1 \leq \delta \leq \max\{k/\alpha,1\}$, such that for any sets $K_1,K_2\subseteq N$ satisfying $|K_1|,|K_2| \leq k$ and $|K_1 \symdiff K_2| \geq \delta$ there exists $Q \in \mathcal{Q}$ that satisfies $|Q\cap K_1| \leq \alpha$, $|Q\cap K_2|\leq \alpha$ and $\feed{\alpha}(Q \cap K_1) \neq \feed{\alpha}(Q \cap K_1)$.
\end{lemma}
\begin{proof}
Assume that $k$ is a power of $2$ (if it is not, we can increase $k$ to the closest power of $2$ without increasing the asymptotic complexity of our sequence). From Lemma~\ref{lem:telescope}, we have that there exists a sequence of length $c k^2 \log(n/k) / (\alpha\delta)$, for some constant~$c$, distinguishing any two sets $K_1,K_2$ satisfying $|K_1| \geq |K_2|$ and $k\geq |K_1| \geq k/2$ and $|K_1 \symdiff K_2| \geq \delta$. We want to show that such a sequence exists for any pair of sets of size at most $k$. We call the sequences from Lemma~\ref{lem:telescope} applied to parameter $k/2^i$ instead of $k$ as $\mathcal{Q}_i$. By concatenating such sequences for $i = 0,1,\dots,\lfloor \log_2 (k/\alpha) \rfloor$ and with sequence $\hat{\mathcal{Q}}$ from Lemma~\ref{lem:telescope2}, we obtain sequence $\mathcal{Q}$ of length:
$
\lceil 3k \log (4n/k) \rceil + \sum_{i = 0}^{\lfloor\log_2 (k/\alpha) \rfloor} c\frac{k^2\log(2^in/k)}{4^i\alpha\delta } \in O((k+\frac{k^2}{\alpha\delta}) \cdot \log(n/k) )$.
Take any two sets $K_1$ and $K_2$ such that $|K_1|, |K_2| \leq k$ and $|K_1 \symdiff K_2| \geq \delta$. Without loss of generality assume that $|K_1| \geq |K_2|$. 
If $|K_1| \leq \alpha$, then the pair is distinguished by $\hat{\mathcal{Q}}$ due to Lemma~\ref{lem:telescope2}. Otherwise,
we find such $i$, that $k2^{-i}\geq |K_1| \geq k 2^{-i-1}$. By Lemma~\ref{lem:telescope}, sequence $\mathcal{Q}_i$ distinguishes $K_1$ from $K_2$. Since $\mathcal{Q}$ contains $\mathcal{Q}_i$ as subsequence, then $\mathcal{Q}$ also distinguishes $K_1$ from $K_2$.
\end{proof}

As mentioned earlier, Theorem~\ref{thm:binary} follows directly from Lemma~\ref{lem:binarydelta} applied for $\delta=1$.
It is worth mentioning that Lemma~\ref{lem:binarydelta}, based on technical development in Lemma~\ref{lem:telescope},  could be seen as more universal tool that could be applied
to the analysis of other feedbacks related to or using parity as its part, c.f., Section~\ref{sec:geberal-feedback}.

\dpa{
\paragraph{Randomized counterpart construction}
\dk{First} 
note that the explicit randomized construction used in the proof of Lemma~\ref{lem:telescope} leads directly to the following corollary:

\begin{corollary}
\label{cor:random1}
There exists an \dk{explicit} randomized algorithm that generates a sequence $\mathcal{Q}$ of $O(\frac{k^2}{\alpha\delta} \cdot \log (n/k))$ sets such that with probability at least $1 - k^2 / 4^{-k}$ \dk{the following holds:} for any sets  $K_1,K_2\subseteq N$ 
\dk{such that}
$k \geq |K_1| \geq k/2$ and $|K_1| \geq |K_2|$ and $|K_1 \symdiff K_2| \geq \delta$, there exists $Q \in \mathcal{Q}$ that satisfies $|Q\cap K_1| \leq \alpha$, $|Q\cap K_2|\leq \alpha$ and 
\dk{$\feed{\alpha}(Q \cap K_1) \neq \feed{\alpha}(Q \cap K_2)$.}
\end{corollary}

A randomized algorithm generating a concatenation of sequences \dk{$\mathcal{Q}_i$, taken} from Corollary~\ref{cor:random1} \dk{for parameters $k/2^i$,} in the same manner as in Lemma~\ref{lem:binarydelta} for $\delta = 1$, results in an explicit randomized algorithm for \gt under \aadv. It is easy to see that if each of $\lfloor \log_2 k \rfloor$ concatenated sequences $\mathcal{Q}_i$ does not fail (\ie it does distinguish all pairs of sets of certain sizes), then the resulting sequence distinguishes all pairs of sets of sizes at most $k$. 

\begin{corollary}
\label{cor:binary_random}
Under $\feed{\alpha}$ feedback and under adaptive \aadv, there exists an explicit randomized solution to \gt with query complexity\\  \dk{$O\left(\left(k+\frac{k^2}{\alpha}\right) \cdot \log \frac{n}{k}\log 1/c\right)$} working with probability at least  $1 - c$, for any 
$c\in (0,1)$.
\end{corollary}
\begin{proof}
The probability that a single of the sequences concatenated in Lemma~\ref{lem:binarydelta} fails to distinguish all sets of certain sizes is at most:
\[
k^2 \cdot 4^{-k} + \sum_{i=0}^{\lfloor \log_2 (k/\alpha)\rfloor} (k/2^i)^2 4^{-(k/2^i)} \leq \frac14 + \sum_{j = 1}^{\infty} j^2 4^{-j} = \frac{107}{108}
\ ,
\]
because $\sum_{j = 1}^{\infty} j^2 4^{-j} = 20/27$. If we concatenate $\lceil \log_{108/107} 1/c \rceil$ independently generated such sequences, we get a sequence that distinguishes all sets with probability at least $1-c$.
\end{proof}
}


\subsection{Maximum expressiveness -- Full feedback}

In this section we consider $\fullfeed{\alpha}$ feedback. Using it, we show that larger expressiveness of feedback allows for smaller query complexity. The following lemma is independent of any feedback function and shows that there exists a query sequence that \emph{$\alpha$-isolates} each element of $K$, in the following sense: for any set $K$ of size at most $k$ and each element in $K$, there exists a query such that this element and at most $\alpha-1$ other elements from $K$ belong to this query.

\begin{lemma}
\label{lem:fullupper}
If $\alpha \geq 9 \log k$ and $k \geq \alpha$, then there exists a sequence $\mathcal{Q}$ of $t=O((k^2 / \alpha^2) \cdot \log n)$ subsets of $N$ with the property that for any set $K \subseteq N$ such that $|K| \leq k$ and any element $x\in K$, there exists  $Q \in \mathcal{Q}$ with the property that $x \in Q$ and $|K \cap Q| \leq \alpha$.
\end{lemma}
\begin{proof}
We prove existence of such family $\mathcal{Q}$ by a probabilistic argument. 
Let  $h =\lceil 16 \log (3n) \rceil$.  The family  $\mathcal{Q}$ consists of 
$t=6h\cdot x \cdot k^2 / \alpha^2$  subsets denoted as $Q_1,\ldots, Q_t$. For each $i=1,\ldots, t$ the set $Q_i$ is generated in the following manner: each element $x\in N$ belongs to $Q_i$ with probability  $\frac{\alpha}{6 \cdot k}$. All the random choices are independent over all elements and subsets.

\subparagraph{Claim 1.} With probability at least $2/3$ each element of $N$ belongs to at least $\frac{ k}{2\alpha} $ queries in the sequence $\mathcal{Q}$.

\textit{Proof of Claim 1.}
For a fixed $x\in N$, let $L_x=|\{i\in [t]: x \in A_i\}|$. 
Clearly, $L_x$ is a sum of  Bernoulli trials and $\E{L_x} = \frac{h x k}{\alpha}$. Due to  the independence of random inclusion of consecutive elements we can use a standard Chernoff bound~\cite{SURV} to get
$
\Pr{L_x < h\frac{k}{2\alpha}} \leq e^{-\frac{hk}{16\alpha}} <\frac{1}{3n}
$,
where the last inequality follows from the fact that, $\frac{hk}{16\alpha} \geq \log (3n)$.
Using the union bound over all $n$ possible elements $v \in N$ we get Claim 1. $\blacksquare$


Consider a sub-sequence of queries from $\mathcal{Q}$, and from all these queries we remove all elements that do not belong to $K$,
namely: $\mathcal{Q}_{K,T}=\{Q_i\cap K: Q_i \in \mathcal{Q} \ \& \ i \in T\}$. 

\subparagraph{Claim 2.} With probability at least $2/3$, for any choice of $K$ with $k$ elements and any  $T$ of 
$\frac{hk}{2\alpha}$ indices, the resulting sequence $\mathcal{Q}_{K,T}$ contains a set with at most $\alpha$ elements.

\textit{Proof of Claim 2.}
Let us fix  any subset  $K\subseteq N$ and a set $T$ with proper cardinalities. In any fixed set $Q^' \in  \mathcal{Q}_{K,T}$ we define its number of elements as $X_{Q^'}$. Clearly, $X_{Q^'}$ is a sum of Bernoulli random variables. 

We have $\E{X_{Q^'}} = \frac{\alpha}{6}$ and by the  Chernoff bound we get
$
\Pr{X_{Q^'} > \alpha} <\Pr{X > 6 \E{X}} < e^{-\alpha} 
$.
%
Due to independence of choices elements in different queries, the probability that the number of elements is greater than $\alpha$ in all $\frac{hxk}{2\alpha}$ sets in $\mathcal{Q}_{K,T}$ is at most
$
e^{-\alpha \cdot \frac{hk}{2\alpha}} =e^{-\frac{hk}{2}} 
$.
%
Recall that the above reasoning was performed for a fixed choice of sets $K$ and $T$. To apply a union bound argument one needs to multiply the above value by the number of all possible choices of sets $K$ and~$T$. The logarithm of the number of possible combinations of $K$ and $T$ equals~to:
\[
\log \left({n \choose k} \cdot {\frac{6h k^2}{\alpha^2}\choose \frac{hk}{2\alpha}}\right) \leq k \log\left(\frac{en}{k}\right) + \frac{hxk}{2\alpha} \log\left(\frac{6 e k}{\alpha}\right),
\]
and since $\alpha > 9 \log k$ we obtain the logarithm of the union-bounded probability multiplied by the number of choices we get:
\begin{align*}
\log \left(e^{-\frac{hk}{2}} \cdot {n \choose k} \cdot {6h \cdot \frac{k^2}{\alpha^2}\choose h \cdot \frac{k}{2\alpha}}\right) &\leq -\frac{hk}{2}  + k \log\left(\frac{en}{k}\right) + \frac{hk }{2\alpha} \log\left(\frac{6 e k}{\alpha}\right)\\
& \leq - \frac{hk}{2}   + k \log\left(\frac{en}{k}\right) + \frac{hk \log k}{18 \log k} \leq -\frac{4hk}{9} + k \log\left(\frac{en}{k}\right) \\
& \leq - 14 k \log\frac{3n}{k} + k\log\frac{en}{k} \leq -13 k\log\frac{3n}{k} < \log\frac{1}{3}
\ .
\end{align*}
Hence,
$
e^{-k \cdot 32 \log(3n/k)} \cdot {n \choose k} \cdot {\frac{6h k^2}{\alpha^2}\choose \frac{hk}{\alpha}} <\frac{1}{3} 
$.
This concludes the proof of Claim 2. $\blacksquare$

Observe that with probability at least $1/3$, a randomly chosen family $\mathcal{Q}$ simultaneously meets conditions described in Claim 1 and Claim 2, by the union bound. Consequently, with probability at least $1/3$, in the randomly generated family $\mathcal{Q}$ for any set $K$ of size $k$ and every sub-sequence $T$ of $\frac{hk}{\alpha}$ queries from $\mathcal{Q}$ there is at least one query $Q_i$ such that $|Q_i \cap K| \leq \alpha$, for some $i\in T$. Hence, such a family $\mathcal{Q}$ exists, by straightforward probabilistic argument. 
Finally, observe that since $\mathcal{Q}$ works for any set $K$ of exactly $k$ elements, then it also does 
for any $K$ such that $|K| \leq k$. 
%
%
%
%
\end{proof}


Interestingly, sequence $\mathcal{Q}$ from Lemma~\ref{lem:fullupper} with parameters $n,k$ does \textbf{not} distinguish all pairs of sets $K_1,K_2$ of size at most $k$. We only know, that each element $x \in K_1$ belongs to some query $Q_{\tau} \in \mathcal{Q}$, with $|K_1 \cap Q_{\tau}| \leq \alpha$. But we may have $|K_2 \cap Q_{\tau}| > \alpha$ and the \aadv may force the feedbacks to be equal on this position for sets $K_1$ and $K_2$. To solve this problem, in the proof of Theorem~\ref{thm:fullupper}, we take the sequence from Lemma~\ref{lem:fullupper} with parameters $n,2k$, and use it for set $K = K_1 \cup K_2$.

\begin{proof}[Proof of Theorem~\ref{thm:fullupper}]

The component $\frac{n}{\alpha}$ follows from the fact, that a simple selector, where each element belongs to one query and each query contains $\alpha$ elements (except the last query that contains at most $\alpha$) has query complexity $O(\frac{n}{\alpha})$ and solves \gt under the $\fullfeed{\alpha}$ feedback and works under \aadv. 
The first part of theorem is a consequence of Lemma~\ref{lem:fullupper}. Specifically, we take the family $\mathcal{Q}$ from Lemma~\ref{lem:fullupper} with parameters $n, 2k$. We observe that for any two sets $K_1$, $K_2$, with $|K_1|, |K_2|\leq \alpha$ and $K_1 \neq K_2$, we have $|K_1 \cup K_2| \leq 2k$ and  $K_1 \symdiff K_2 \neq \emptyset$. Take any $x \in K_1 \symdiff K_2 $ and observe that due to Lemma~\ref{lem:fullupper} there is a query $Q\in \mathcal{Q}$ such that $x\in Q$ and $|Q\cap (K_1 \cup K_2)| \leq \alpha$. Hence, $|Q\cap K_1| \leq \alpha$, $|Q\cap K_2| \leq \alpha$ and $\fullfeed{\alpha}(Q\cap K_1) \neq \fullfeed{\alpha}(Q\cap K_2)$. Consequently, $\mathcal{Q}$ solves \gt under \aadv, by Proposition~\ref{fct:technique}.

The second part of the theorem follows from the fact that we can use the result from Theorem~\ref{thm:binary} and obtain a sequence of length $O(\frac{k^2}{\alpha} \log(n/k))$ (this results does not require the assumption on $\alpha$ and also works under \aadv). Note that we do not need the $O(k \log(n/k))$ component here because under $\fullfeed{\alpha}$ in the case where $k \leq \alpha$, the problem is solvable using a single query.
 \end{proof}

\dpa{
\paragraph{Randomized counterpart construction}
In the proof of Lemma~\ref{lem:fullupper} we construct a sequence at random and show that it satisfies a certain condition with probability at least $1/3$. Clearly, from this we can obtain an explicit randomized construction that succeeds with probability~$1/3$, \dk{and by iterating it a certain number of times we get the following:}

\begin{corollary}
Under $\fullfeed{\alpha}$ feedback and under adaptive \aadv, there exists an explicit randomized solution to \gt with query complexity 
\begin{align}\nonumber
& \dk{O\left(\frac{k^2}{\alpha\beta}\left(\frac{\beta}{\alpha} + \log n\right)\cdot \log n \cdot \log 1/c\right)} & \text{if } \alpha > 18 \log k, \\\nonumber
& O\left(\frac{k^2}{\alpha} \cdot \log\frac{n}{k} \dk{\cdot \log 1/c}\right) & \text{otherwise}.
\end{align}
 working with probability at least  $1 - c$, for any 
$c\in (0,1)$.
\end{corollary}
\begin{proof}
We first observe that if $k\leq \alpha$ then, because of the $\fullfeed{\alpha}$ feedback, a single query containing all elements from set $N$ solves \gt. Hence, we focus on case $k \geq \alpha$. If $\alpha \leq 18 \log k$, then we can use the result from Corollary~\ref{cor:binary_random} and obtain a \dk{desired} sequence of length $O((k + \frac{k^2}{\alpha}) \log(n/k)\dk{\log 1/c})$ \dk{with probability of success at least $1-c$}, which becomes $O(\frac{k^2}{\alpha} \log(n/k)\dk{\log 1/c})$, because $k \geq \alpha$. Finally if $18\log k < \alpha \le k$, we can use the construction from the proof of Lemma~\ref{lem:fullupper} that fails with probability at most $1/3$. By repeating it $\lceil \log_3 1/c\rceil$ times, \dk{independently, and concatenating the resulting sequences} we obtain a desired probability of success.
\end{proof}
}


\subsection{General feedback}
\label{sec:geberal-feedback}

In our construction of $\abfeed{\alpha,\beta}(X)$ (introduced in Definition~\ref{def:feedbacks}) we use the following 
code, where notation $\bigoplus S$ denotes bit-wise XOR of all the elements of set $S$. Such a code was defined in \cite{Censor-HillelHL15} and its explicit construction can be found in \cite{roth2006introduction}.
\begin{definition}
\label{def:BCC}
 An $[n, \beta, \gamma]$-BCC-code is a set $C\subseteq \{0,1\}^\beta$ of size $|C| = n$ such that for any two subsets $S_1,S_2 \subseteq C$ (with $S_1\neq S_2$) of sizes $|S_1|,|S_2| \leq \gamma$ it holds that $\bigoplus S_1 \neq \bigoplus S2$.
\end{definition}

\begin{lemma}{\cite[Lemma 2]{Censor-HillelHL15}} 
\label{lem:bcc}
There exist $[n,\beta , \gamma]$-BCC codes with $\beta = O(\gamma \log n)$. 
\end{lemma}
\begin{proof}[Proof of Theorem~\ref{thm:generalupper}]
First we prove the part of the theorem that works under the assumption $\alpha > 18 \log k$. We denote $\beta^' =\min\left\{\left\lfloor \frac{\beta - 1}{c \log n}\right\rfloor,\alpha \right\}$. We note that if $\beta < \log n$, then we get $\beta^' = 0$ but the result follows simply by choosing $\mathcal{Q}$ as the sequence from Theorem~\ref{thm:binary}. 
Note that we must have $\beta^' \leq \alpha$ because input set $X$ cannot contain more than $\alpha$ elements.

Assume that $\beta > \log n$ and let us take the family  from Lemma~\ref{lem:binarydelta} with parameter $\delta = \beta^'$ and concatenate it with the family from Lemma~\ref{lem:fullupper} with parameters $2k$ and $n$
. Observe, that this resulting family $\mathcal{Q}$ (composed of two parts $\mathcal{Q}_1$ and $\mathcal{Q}_2$) has length $ O\left(\frac{k^2}{\alpha\beta^'} \cdot \log n\right) = O\left(\max\left\{\frac{k^2}{\alpha^2}\log n,\frac{k^2}{\alpha \beta}\log^2 n\right\}\right)$. We will show that this family distinguishes under $\abfeed{\alpha,\beta}$  feedback model, any two sets $K_1$, $K_2$ satisfying $|K_1|,|K_2| \leq k$.
We will consider two cases.

In the case, where $|K_1 \symdiff K_2| < \beta^'$, we pick an arbitrary element $s \in K_1 \symdiff K_2$. Without loss of generality assume that $s \in K_1$. By Lemma~\ref{lem:fullupper} in some $Q \in \mathcal{Q}_2$ we have $|Q\cap (K_1 \cup K_2)| \leq \alpha$ and $s \in Q$. 

We want to show that:
\[
\bigoplus_{s \in Q \cap K_1} \bcc(s)  \neq \bigoplus_{s \in Q\cap K_2} \bcc(s).
\]
We denote $K_1^' = (K_1 \setminus K_2) \cap Q$ and $K_2^' = (K_2 \setminus K_1) \cap Q$ and $T = K_1\cap K_2 \cap Q$. We know that $K_1^' \neq \emptyset$, because $s \in K_1^'$ and also $K_1^' \neq K_2^'$ because $s \notin K_2^'$. Since $|K_1 \symdiff K_2| < \beta^'$, and $K_1^',K_2^' \subset K_1 \symdiff K_2$ then also $|K_1^'|,|K_2^'| <\beta^'$. By the definition of BCC codes~\cite[Lemma 2]{Censor-HillelHL15} we have that:
$
\bigoplus_{s \in K_1^' } \bcc(s)  \neq \bigoplus_{s \in K_2^' } \bcc(s).
$
Using the properties of operation XOR:
\[
\bigoplus_{s \in Q\cap K_1} \bcc(s) = \bigoplus_{s \in K_1^' } \bcc(s) \oplus \bigoplus_{s \in T} \bcc(s) \neq \bigoplus_{s \in K_2^' } \bcc(s) \oplus \bigoplus_{s \in T} \bcc(s) = \bigoplus_{s \in Q\cap K_2} \bcc(s).
\]

Hence, if $|K_1\symdiff K_2| \leq \beta^'$ then there exists $Q$ such that $|Q\cap (K_1\cup K_2)|\leq \alpha $. Consequently, $|Q\cap K_1|\leq \alpha$ and $|Q\cap K_2|\leq \alpha$.
We have obtained that $\abfeed{\alpha,\beta}(Q\cap K_1) \neq \abfeed{\alpha,\beta}(Q\cap K_2)$, thus $Q$ distinguishes $K_1$ and $K_2$.

If $|K_1 \symdiff K_2| \geq \beta^'$, then by Lemma~\ref{lem:binarydelta} the first part of our sequence $\mathcal{Q}_1$ distingushes $K_1$ and $K_2$ under the binary feedback. Since $\abfeed{\alpha,\beta}$ includes the binary feedback, our sequence $\mathcal{Q}$ distinguishes $K_1$ from $K_2$. 

After considering both cases, we have that sequence $\mathcal{Q}$ distinguishes any two sets $K_1,K_2$ of size at most $\alpha$ and thus we can use Proposition~\ref{fct:technique} and obtain  that $\mathcal{Q}$ solves the $(n,k)$-Group-Testing problem under \aadv.



The second part of the theorem follows from the fact that we can use the result from Theorem~\ref{thm:binary} and obtain a sequence of length $O(\frac{k^2}{\alpha} \log(n/k))$ (this results does not require the assumption on $\alpha$). 
\end{proof}

\dpa{
\paragraph{Randomized counterpart construction}
\dk{Unlike in previous sections, for $\abfeed{\alpha,\beta}(X)$ feedback it is not simple to provide an explicit algorithm, even randomized.
This is because} an explicit (even randomized) construction of BCC-codes is not known. 
\dk{Therefore, we suggest this problem as one of interesting open directions.}
}

%
\section{Lower bound}
\label{sec:lower}

\begin{proof}[Proof of Theorem~\ref{thm:fulllower}]
An $(n,k,k)$-selector is a sequence of queries $Q_1,Q_2,\dots,Q_t$ such that for any $|K| \leq k$ and any element $x \in K$, for some query $Q$ we have $Q\cap K = \{x\}$. An $(n,k,k)$-selector is known to have query complexity $\Omega(\min\{n, (k^2/\log k) \cdot \log n\})$~\cite{ClementiMS01} and 
is known to exist with query complexity $O(k^2 \log n)$~\cite{erdos1985families}. Assume that $n$ is sufficiently large and let $c_1$ be such a constant that $(n,k,k)$-selector of query complexity $t_1 \leq c_1 (k^2/\log k) \cdot \log n $ does not exist; $c_1$ is well-defined by~\cite{ClementiMS01}. Let $c_2$ be a constant such that $(n,\alpha+1,\alpha+1)$-selector of size $t_2 \leq c_2 \alpha^2 \log n$ exists;
$c_2$ is well-defined by~\cite{erdos1985families}. Let $\mathcal{R}= \langle R_1, R_2, \dots, R_{t_2}\rangle$ be such a selector.

The proof of the theorem is by contradiction. Assume that there exists a sequence $\mathcal{Q} = \langle Q_1,Q_2,\dots, Q_t\rangle$ solving \gt with some $(\alpha,\beta)$-feedback function, such that the length of the sequence is $t \leq \frac{c_1}{c_2}\cdot \frac{k^2}{\alpha^2} \log^{-1} k$.

First we show the following fact: for any set $K \subset N$ of size $|K|\leq k$ and any element $x \in K$ there exists a set $Q$ in sequence $\mathcal{Q}$ such that $|K\cap Q| \leq \alpha + 1$. Assume the contrary and fix $K$ and $x$ that violate the fact. Observe that sets $K$ and $K\setminus \{x\}$ may produce the same feedback for any $(\alpha,\beta)$-feedback function (regardless of the value of $\beta$). This is because for any $Q_i$ such that $x\in Q_i$ we have $|Q_i \cap K| \geq \alpha + 2$ and $|Q_i \cap (K\setminus \{x\})| \geq \alpha + 1$. Hence, in the case of $Q_i \cap K$ the \hadv may provide to the feedback function the same $\alpha$ elements as in $Q_i \cap (K\setminus \{x\})$. Since the feedback function is deterministic, the results will be the same, which is a contradiction with the fact that $\mathcal{Q}$ solves the \gt problem.

Next we transform $\mathcal{Q}$ into an $(n,k,k)$-selector: we take the $(n,\alpha+1,\alpha+1)$-selector $\mathcal{R}= \langle R_1, R_2, \dots, R_{t_2}\rangle$ and construct a sequence 
$\mathcal{S} = \langle Q \cap R \text{, for each } R = R_1, R_2, \dots, R_{t_2}   \text{, for each } Q = Q_1, Q_2, \dots, Q_t \rangle$.
The obtained family $\mathcal{C}$ has $t \cdot t_2 \leq t_1$ queries.
We prove that it is also an $(n,k,k)$-selector.
For any set $K$ with $|K| \leq k$ and any $s \in S$ there exists, by the property of the sequence $\mathcal{Q}$ proved above, a set $Q$ in sequence $\mathcal{Q}$ such that $|Q \cap K| \leq \alpha + 1$. Now, since $\mathcal{R}$ is an $(n,\alpha+1,\alpha+1)$-selector, there exists a set $R$ in $\mathcal{R}$ such that $R\cap (Q \cap K) = \{x\}$. By the construction of $\mathcal{S}$, set $Q\cap R$ belongs to sequence $\mathcal{S}$, hence element $x$ is selected by family $\mathcal{S}$. Hence $\mathcal{S}$ is an $(n,k,k)$-selector. We know however that $(n,k,k)$-selector of length at most $t_1$ does not exist, and thus obtain a contraction showing that such a family $\mathcal{Q}$ cannot exist.
\end{proof}
%
\dpa{
\paragraph{Randomized counterpart \dk{result}}
In Theorem~\ref{thm:fulllower} we show that any sequence that solves \gt must have length of at least $\Omega\left(\frac{k^2}{\alpha^2} \log^{-1} k\right)$. Thus, a randomized algorithm generating sequences that solve \gt with at least a constant probability must have expected query complexity of $\Omega\left(\frac{k^2}{\alpha^2} \log^{-1} k\right)$, \dk{since each correct sequence must have such length (by Theorem~\ref{thm:fulllower}).}
\begin{corollary}
If $n > k^2\log n/\log k$, then any randomized solution to \gt under any $(\alpha,\beta)$-feedback has expected query complexity 
$\Omega\left(\frac{k^2}{\alpha^2} \log^{-1} k\right)$ for some adaptive \hadv.
\end{corollary}
}
\section{Some $(\alpha,\beta)$-feedbacks are better than others}
\label{sec:case-study}
\dk{One could be tempted to develop a similar universal reduction, \dk{as in the proof of the lower bound in Section~\ref{sec:lower}, also} for upper bounds -- between a setting with any $(\alpha,\beta)$-feedback 
and some strong selectors. However, in this section we show that such a reduction does not exist: we define two seemingly very similar feedback functions with the same values of $\alpha, \beta$ and show that the resulting query complexities for these two feedbacks are asymptotically very different.}

Consider the following two feedback functions, both being $(\alpha,\beta)$-feedbacks for any $\alpha\le k$ and $\beta = 2 \lceil\log_2 n \rceil$. In the following we associate each element with its identifier. We assume that each identifier has exactly $\lceil \log_2 n \rceil$ bits and is different from the string of~all~zeros.
\[
F_1(S) =
\begin{cases}
(\min S) \bigparallel (\min(S \setminus \{\min S\}), \quad \text{if }2\leq |S| \leq \alpha \\
(\min S) \bigparallel 00\dots 0, \quad \text{if } |S| = 1~. \\
\end{cases}
\]
\[
F_2(S) =
\begin{cases}
(\min S) \bigparallel  \min(S \setminus \{\min S\}), \quad \text{if }2\leq |S| \leq \alpha \text{ and $|S|$ is odd} \\
(\min(S \setminus \{\min S \})) \bigparallel  \min S , \quad \text{if }2\leq |S| \leq \alpha  \text{ and $|S|$ is even} \\
(\min S) \bigparallel 00\dots 0, \quad \text{if } |S| = 1~. \\
\end{cases}
\]
We show that the query complexity of \gt with feedback $F_2$ is substantially lower than with feedback $F_1$.

\begin{corollary}
\label{cor:upper-2min}
For any $\alpha \leq k$, query complexity of \gt under feedback $F_2$ under \aadv is $O\left(\frac{k^2}{\alpha} \log (n/k) \right)$.
\end{corollary}

\begin{proof}
We can see that under feedback $F_2$ we can deduce the parity bit from the feedback from every request with intersection at most $\alpha$, by checking the order of the two outputted elements (note that if there is only one or no outputted elements, the parity is obvious). Hence, the corollary is a direct consequence of Theorem~\ref{thm:binary}.
\end{proof}
%
\paragraph{Remark}
An unexpected inspiration for the feedback $F_2$ is a type of move, called \emph{count signal}, used in contract bridge. Contract bridge is a card game, where players play in pairs (but without seeing non-revealed cards of other players) and sometimes it is crucial to exchange some information between the partners about their cards. The only way to disclose is by revealing (playing) the cards, but the \emph{order} in which the cards are played can have some meaning. A count signal is exactly the $F_2$ feedback, where the parity of one player's cards (in some particular suit) is disclosed by the order in which he/she plays the cards. For example a player holding $\diamondsuit Q 963$ (meaning Queen, 9, 6, 3 in diamonds) plays $6$ and then $3$ to show even number of cards in diamonds. 

On the other hand, returning the same two minimal values as in $F_2$ but always in order, results in a dramatic increase in the query complexity, \dk{even under Honest Adversary.}

\begin{theorem}
\label{thm:lower-2min}
For any $k, \alpha$, query complexity of \gt under feedback $F_1$ is $\Omega(\min\{n,k^2\})$ under \hadv.
\end{theorem}
\begin{proof}
Assume, that we have a sequence of queries solving \gt. Let us denote the queries by $Q_1,Q_2,\dots, Q_t$, where $t$ is the number of queries. Recall that $N$ denotes the set of all the elements.

If $k \geq n/2$, then assuming that we had $t < k / 2$, take an arbitrary set of $k$ elements $K \subset N$. And observe that feedback to each query reveals at most two identifiers. The total number of identifiers revealed is $t \cdot 2 < k$. Hence there is an element $x \in K$ that is never returned by the feedback. It is easy to see that by the properties of the feedback function $F_1$, the feedbacks for set $K \setminus \{x\}$ would be identical as for set $K$ for each query. Hence if $k \geq n/2$ we must have $t \geq k/2 \geq n/4$.

Let us now consider the more interesting range of $k < n/2$. We define sets of indices $T_{>1} = \{\tau :  |Q_\tau| > 1\}, T_{= 1} = \{\tau :  |Q_\tau|  = 1\}$. Denote the following set of elements:
\[
M = \bigcup_{\tau \in T_{=1}} Q_\tau \cup \bigcup_{\tau \in T_{>1}}\left( \{ \min Q_{\tau} \} \cup \{ \min (Q_{\tau} \setminus \min Q_{\tau})\} \right)
\ .
\]
We know that $|M| \leq 2t$, because we take at most $2$ elements from each query. Denote set $R= N \setminus M$. Set $R$ are the elements from $N$ that are not smallest (or second smallest) in any of the queries.  

If we have $|R| < k$, then $n - 2t < k$ and since $k \leq n/2$ we have $t \geq n/4$. 

Assume that $|R| \geq k$, consider $R$ ordered in the decreasing order of identifiers. Denote this ordering as $r_1,r_2,\dots$ and let $R_{i,j} = \{r_i,r_2,\dots, r_j\}$. Denote the indices of queries that include element $r_i$ as $T^{(r_i)}_{> 1} = \{\tau \in T_{>1} | r_i \in Q_\tau\}$.

For any $j \in N$ and $i=1,2,\dots,j$, define two sets of query indices:
\[
A_j(i) = \{\tau \in T^{r_i}_{>1}, |Q_{\tau} \cap R_{i+1,j}| = 0 \}~,
\]
\[
B_j(i) = \{\tau \in T^{r_i}_{>1}, |Q_{\tau} \cap R_{i+1,j}| = 1 \}~.
\]
We will prove the following: \\
\textit{Claim 1: For any $j,i$ we have $2 \cdot |A_j(i)| + |B_j(i)| > k - j$.}

Assume on the contrary that this does not hold for some particular $j,i$ and observe that then we can find for every query $Q_\tau$ for $\tau \in A_j(i)$ two elements that belong to $Q_\tau$ and are smaller than $r_i$. Take such two elements for each $\tau \in A_j(i)$. We have $2|A_j(i)|$ elements, call this set $A$. For every $\tau \in B_j(i)$ find one element that belongs to $Q_\tau$ and is smaller than $r_i$. We take $|B_j(i)|$ such elements (one for each of $B_j(i)$) and call this set $B$. Consider two sets 
\[
S = R_{1,j} \cup A \cup B~,
\]
\[
S^' = R_{1,j} \cup A \cup B \setminus \{r_i\}~.
\]
Observe that $|S| \leq |S^'| \leq  j + k -j = k$. We will compare the feedbacks for $S$ and $S^'$ and show that the feedbacks are identical for each query.
Note that for every $\tau \in T^{(r_i)}_{> 1}$, $S^' \cap Q_\tau$ has at least two elements that are smaller than $r_i$. Hence if $|S \cap Q_\tau| \leq \alpha$ then surely $F_1(S^' \cap Q_\tau) = F_1(S \cap Q_\tau)$. If $|S \cap Q_\tau| > \alpha$, there is a simple strategy of an adversary to ensure equal feedbacks. The adversary selects an arbitrary set $X$ with $|X| = \alpha$ satisfying, $X \subset S^' \cap Q_\tau$ and  $X \subset S \cap Q_\tau$ and passes $X$ to the feedback function. Hence the feedback in step $\tau$ is identical for both $S$ and $S^'$. Note that, since $r_i \notin M$, then $r_i$ does not belong to any other query than the queries with indices in $ T^{(r_i)}_{> 1}$, hence we cannot distinguish $S$ from $S^'$. This means that the query sequence does not solve the set learning problem. We obtained a contradiction, which proves the claim.
\\
In the next claim we prove that sets $A_j(i)$ and $B_j(i)$ are disjoint.
\\
\textit{Claim 2: For any $j$, we have $A_j(i) \cap A_j(i^') = \emptyset$ and $B_j(i) \cap B_j(i^') = \emptyset$, for $i,i^' \leq j$ and $i \neq i^'$.}

Assume on the contrary that for some $j$ and $i,i^' \leq j$ we have $A_j(i) \cap A_j(i^') \neq \emptyset$ and take arbitrary $\tau^* \in A_j(i) \cap A_j(i^')$. Assume without loss of generality that $i < i^'$. By the definition of sets $A$ we have $r_i \in Q_{\tau^*}$ and $r_{i^'} \in Q_{\tau^*}$. Since $i<i^' \leq j$ we also have $r_{i^'} \in R_{i+1,j}$, hence $|Q_{\tau^*} \cap R_{i+1,j}| \geq 1$ and $\tau^* \notin A_j(i)$ a contradiction. 
Now if $B_j(i) \cap B_j(i^') \neq \emptyset$ then similarly take $\tau^* \in B_j(i) \cap B_j(i^')$ and assume $i < i^'$. We have $r_i, r_{i^'} \in Q_{\tau^*}$. We know that $|Q_{\tau^*} \cap R_{i^' +1,j}| = 1$ thus $|Q_{\tau^*} \cap R_{i^',j}| = 2$. Which implies that $|Q_{\tau^*} \cap R_{i +1,j}| \geq 2$ and $\tau^* \notin B_j(i)$. We obtained a contradiction proving the claim.

We fix $j^* = \lfloor k/2 \rfloor$. From Claim 1 we have $2 \cdot |A_{j^*}(i)| + |B_{j^*}(i)| > k / 2$ for each $i = 1,2,\dots,j$. Sets $A_j(i)$ contain indices of queries hence using Claim 2 we get: 

\[
t \geq \left | \bigcup_{i =1}^{j^*} A_{j^*}(i) \right | =  \sum_{i =1}^{j^*} \left |A_{j^*}(i) \right |~,
\]
\[
t \geq \left | \bigcup_{i =1}^{j^*} B_{j^*}(i) \right | =  \sum_{i =1}^{j^*}  \left |B_j(i) \right |~.
\]
Adding up the above inequalities gives us:
\[
3t \geq 2\cdot \sum_{i =1}^{j^*}  \left |A_{j^*}(i) \right | +   \sum_{i =1}^{j^*} \left |B_{j^*}(i) \right | = \sum_{i=1}^{j^*} (2 |A_{j^*}(i) | + \left |B_{j^*}(i) \right |)  \geq j^* \cdot k / 2 \geq k^2 / 4 - k/2
\ .
\]
Thus, finally we get $t \geq k^2 / 12 - k/6$.
%
\end{proof}

\dpa{
\paragraph{Randomized counterpart \dk{result}}
In Theorem~\ref{thm:lower-2min} we show that any sequence that solves \gt under feedback $F_1$ under \hadv must have length of at least $\Omega\left(\min\{n,k^2\}\right)$. Thus, a randomized algorithm generating sequences that solve \gt under this feedback with at least a constant probability must have expected query complexity of $\Omega\left(\min\{n,k^2\}\right)$, \dk{since each correct sequence must have such length (by Theorem~\ref{thm:lower-2min}).}
\begin{corollary}
For any $k, \alpha$, query complexity of any randomized solution to \gt under feedback $F_1$ is $\Omega(\min\{n,k^2\})$ under adaptive \hadv.
\end{corollary}
}

%
\section{Discussion of results and open directions}
\label{sec:future}
We conclude the paper with four promising future directions. 
\paragraph{Sparsity}
In addition to the query complexity, there are two additional metrics of \grouptesting solutions that are 
\dk{studied in
}
literature. These parameters are: the maximum number of queries to which 
\dk{an
element belongs to (typically denoted by $w$) 
and the maximum size of a query (typically denoted by $\rho$).} The interplay between all these three parameters, \dk{i.e., query complexity, $w$ and $\rho$,} was carefully studied in~\cite{Inan2020}  in case of the Beeping feedback,
\dk{and in some other recent works \cite{HKK20,Johnson2020} the sparsity of some particular selectors was established and discussed.}
It is possible to derive bounds on parameters $w$ and $\rho$ also for the query sequences considered in this paper. In particular, the sequence in Theorem~\ref{thm:binary} under the $\feed{\alpha}$ feedback has $w = O(k \log(n/k))$ and $\rho = O(n\alpha / k)$, where both bounds can be obtained by a small modification of the analysis in Section~\ref{sec:binary}. An interesting future direction would be to study tradeoffs between query complexity and the values of $w,\rho$ for different feedback models, \dk{in particular, for different capacity $\alpha$ and expressiveness $\beta$.} 


\paragraph{Randomness}
A popular line of research in \grouptesting is to consider randomized solutions~\cite{coja2020information, Bondorf21, Johnson2020}. While in this work we focus on deterministic solutions, some of our algorithms 
\dk{have their simply constructed randomized counterparts, also presented in this work.}
Randomized algorithms defined in this way 
\dk{correctly distinguish}
{\em all} sets~$K$.
This can be contrasted with existing solutions \dk{that, typically, have weaker guarantees: with some probability, to correctly distinguish a} 
{\em randomly chosen} set~$K$ \dk{from other sets of size at most~$k$, or to correctly identify each element only with some probability (resulting in some false-positives and/or false-negatives with non-zero probability)}.
\dk{Moreover, they typically work against a weaker non-adaptive version of an adversary, who has to choose the unknown set~$K$ before the random choices of the algorithm.}
An interesting future direction would be to 
investigate 
\dk{how different probabilistic guarantees and types of adversaries influence the
query complexity of generalized Group Testing.} 
\dk{Another intriguing question is how random perturbations of the feedback function (see e.g.,~\cite{Scarlett2020}) affect the query complexity.
Finally, designing efficient coding (i.e., constructing queries) and decoding (i.e., reconstructing set $K$ from the feedback) algorithms, working in polynomial time, is a challenging open direction, sometimes even for randomized algorithms (c.f., Section~\ref{sec:geberal-feedback} with $\abfeed{\alpha,\beta}(X)$ feedback).}


\paragraph{Other feedbacks}

The third direction, motivated by subtle examples of the considered $(\alpha,2\lceil\log_2 n\rceil)$-feedbacks of different
query complexity in Section~\ref{sec:case-study}, is to study other specific well-motivated classes of $(\alpha,\beta)$-feedbacks and their complexities.
Although \dk{all $(\alpha,\beta)$-feedbacks} have to observe the universal lower bounds, \dk{such as the one in Theorem~\ref{thm:fulllower},}
their actual query complexity might be asymptotically larger.

\paragraph{Other adversaries}
Observe that in our proofs of the lower bounds, Theorems~\ref{thm:fulllower} and \ref{thm:lower-2min}, we use a weak \hadv. This makes our lower bounds stronger and suggests that in case of deterministic non-adaptive algorithms, the adversary that uses \dk{{\em some}} fixed function $\adv$ \dk{may have} similar power to the one being allowed to return arbitrary subsets. 
\dk{What actually follows from our results is that this adversarial impact may be similar for the best feedbacks in the class of $(\alpha,\beta)$-feedbacks, but does not necessarily tell us about the impact for a specific feedback function.} 
This opens an interesting direction of studying the impact of 
adversarial power, \dk{and more generally non-adaptiveness and ``maliciousness'',} to the \grouptesting problem, \dk{not only for general classes of $(\alpha,\beta)$-feedbacks (universal lower bounds, matching by upper bounds obtained for some $(\alpha,\beta)$-feedbacks), but also for specific well-motivated feedback functions.}

\bibliographystyle{abbrv}
\bibliography{biblio}

\newpage

\appendix


\section{Auxiliary tools}
\subsection{Proof of Proposition~\ref{fct:expres}}
\begin{proof}
There are at most $2^{\bar{\alpha}}$ possible inputs to function $f_1$. Hence, since it is deterministic, there are at most $2^{\bar{\alpha}}$ possible outputs. Take the family of all subsets of $N$ of size at most $N$ and define the partition of this family into subfamilies -- \dk{each consisting of sets}
with 
the same value of $f_1$. This partition has at most $2^{\bar{\alpha}}$ elements, \dk{because this is the size of the domain of $f_1$}. Fix an arbitrary ordering of this partition and enumerate its elements. Each subfamily receives a unique label with at most $\bar{\alpha}$ bits. Let feedback function $f_2$ for each set $K$ return the label of the subfamily to which $K$ belongs. Such feedback function has expressiveness at most $\bar{\alpha}$ and clearly it 
satisfies the property required from $f_2$ in the statement of the fact.
\end{proof}

\subsection{Proof of Proposition~\ref{fct:technique}}
\begin{proof}
Let $\mathcal{S}_{adv}$ denotes the set of all strategies of \aadv. From the definition of the adversary, we have $\adv(Q_{\tau} \cap K_1, \tau) = Q_{\tau} \cap K_1$ and $\adv(Q_{\tau} \cap K_2, \tau) = Q_{\tau} \cap K_2$. Hence position $\tau$ of feedback vector $\mathcal{F}(K_1,\adv)$ equals to $\gfeed(Q_{\tau} \cap K_1)$ for any strategy of the adversary. 
Similarly position $\tau$ of feedback vector $\mathcal{F}(K_2,\adv)$ equals to $\gfeed(Q_{\tau} \cap K_2)$ for any strategy of the adversary. Since $\gfeed(Q_{\tau} \cap K_1) \neq \gfeed(Q_{\tau} \cap K_2)$, then 
\dk{$\{\mathcal{F}(K_1,\adv) : \adv \in \mathcal{S}_{adv}(\cQ,K_1)\} \cap  \{\mathcal{F}(K_2,\adv) : \adv \in \mathcal{S}_{adv}(\cQ,K_2)\} = \emptyset$.}
\end{proof}

\subsection{Proof of Fact~\ref{clm:oddeven}}
\begin{proof}
We have:
\begin{align*}
(1-2p)^n = ((1-p) - p)^n &= \sum_{k=0}^n {n \choose k} (-p)^k (1-p)^{n-k} \\
& = \sum_{k = 0}^{\lceil n/2 \rceil} {n \choose 2k}p^{2k}(1-p)^{n-2k} - \sum_{k = 0}^{\lceil n/2 \rceil} {n \choose 2k+1}p^{2k+1}(1-p)^{n-2k-1} \\
& = \Pr{X \text{ is even}} - \Pr{X \text{ is odd}}. 
\end{align*}
An since $1= \Pr{X \text{ is even}} + \Pr{X \text{ is odd}} $, we get $\Pr{X \text{ is odd}} = (1 - (1-2p)^n)/2$.
\end{proof}


\end{document}